\documentclass[11pt]{article}
\usepackage{axient_ffra_compat}
\usepackage{algorithm,algpseudocode,placeins,url,seqsplit}
\hypersetup{colorlinks=true,linkcolor=black,citecolor=black,urlcolor=blue!55!black,
  pdftitle={Axient: Manifest-Bound Evidence for On-Chain Financial Protocols: Seven-Layer Derivation, Correlation, Tamper Rejection, and Reproducible Claim Promotion},
  pdfauthor={Maksym Nechepurenko},
  pdfsubject={Manifest-bound evidence and reproducible claim promotion for hybrid financial protocols},
  pdfkeywords={end-to-end evidence, provenance, reproducibility, smart contracts, financial protocols, manifest binding, tamper rejection, replay, auditability}}
\selectfont
\newcolumntype{P}[1]{>{\raggedright\arraybackslash}p{#1}}
\crefname{definition}{definition}{definitions}
\crefname{assumption}{assumption}{assumptions}
\crefname{proposition}{proposition}{propositions}
\crefname{theorem}{theorem}{theorems}
\crefname{corollary}{corollary}{corollaries}
\crefname{remark}{remark}{remarks}
\definecolor{axblue}{HTML}{4263EB}
\definecolor{axteal}{HTML}{14B8A6}
\definecolor{axnavy}{HTML}{0B172A}
\definecolor{axgray}{HTML}{EFF3F8}
\newcommand{\Hh}{\mathsf{H}}
\newcommand{\interimnote}{\begin{center}\fcolorbox{axteal!70!black}{axteal!6}{\parbox{0.93\textwidth}{\centering\textbf{Registered local evidence result.} Twelve uniquely correlated FIA records were retained with 84 \texttt{ACTUAL\_EVIDENCE} positions, two matching clean materializations, a deterministic strict archive accepted by a separate verifier, and a selected ordered replay bound to the archive.}}\end{center}}

\title{\textbf{Axient: Manifest-Bound Evidence for On-Chain Financial Protocols: Seven-Layer Derivation, Correlation, Tamper Rejection, and Reproducible Claim Promotion}}
\author{Maksym Nechepurenko\\\small Research Department of Devnull FZCO, Dubai, UAE\\\small \href{mailto:maksym@devnull.ae}{maksym@devnull.ae} \;|\; \href{https://orcid.org/0000-0002-9515-8841}{ORCID: 0000-0002-9515-8841}}
\date{30 September 2026}
\ffraSetup{series={Axient Research \textperiodcentered\ Paper 7},author={Maksym Nechepurenko},affiliation={Research Department of Devnull FZCO, Dubai, UAE},email={maksym@devnull.ae},orcid={0000-0002-9515-8841},version={r0.6.4},date={30 September 2026}}

\begin{document}
\maketitle

\begin{abstract}
Hybrid on-chain financial protocols are frequently evaluated with evidence that is individually useful but collectively insufficient: a unit test, transaction receipt, screenshot, service response, or several matching hashes may be presented as proof of a complete workflow even when the layers share one generated source or omit the financial authority that matters. This paper develops a manifest-bound evidence architecture for such systems. A sealed execution manifest binds declared system identities, roles, schemas, and release identity. Seven evidence layers are retained with registered derivations and linked by a common correlation identity. A conjunctive claim-promotion rule requires provenance, predicate-specific root conditions, negative controls, reproducible packaging, historical replay, and independent review in addition to hash equality. We formalize the evidence graph, prove that copied hashes cannot establish independent derivation, show how collision-resistant commitments make specified tampering detectable, define exactly-once financial-effect evidence, and distinguish original historical identities from later reproducibility attestations. The registered local case study records twelve parent Financial Interaction Assertions in one accepted cohort: 84 retained \texttt{ACTUAL\_EVIDENCE} positions, two clean materializations with one byte-identical file index, a deterministic strict archive accepted by a separate verifier, and a selected ordered replay bound to that archive. FIA-011 is conjunctive and retains both stale-before-admission and book-epoch-change-mid-position subcases. The result demonstrates a reviewable chain of custody from financial assertion to replay-bound package while preserving the difference between observed evidence, interpretation, and future hypotheses.
\end{abstract}


\section{Introduction}
A complex financial protocol spans more than a blockchain. A user acts through an interface; an API authenticates and validates the command; a durable service orchestrates local or external execution; a smart contract changes financial state; indexers and journals reconstruct the result; and an audit surface presents a human-readable trace. The phrase ``end-to-end'' suggests that the same action has been exercised and observed across the full chain. Yet the evidence commonly used to support that claim is heterogeneous and easy to overstate.

A transaction receipt establishes that one transaction was mined. It does not establish that the intended user command produced it, that the called selector belonged to the approved deployment, that the service projection was derived from the receipt, or that the audit interface did not display fixture data. A screenshot establishes that pixels were rendered. It does not establish the state source. Seven equal hashes may still arise from one generator copied into seven files. A green CI run confirms only the bounded task encoded in that workflow. These are not failures of the artifacts themselves; they are failures of claim typing.

The Axient programme treats evidence as a first-class system. Its public evidence companion records one twelve-parent registered local cohort, 84 retained evidence positions, two matching clean materializations, a deterministic strict archive accepted by a separate verifier, and a selected ordered replay bound to that archive \citep{axientfia2026}. The present paper generalizes the resulting architecture beyond event markets. Its object is a manifest-bound proof chain for hybrid financial protocols.

\paragraph{Contributions.}
\begin{enumerate}[label=(C\arabic*)]
\item A typed evidence model that distinguishes observed, reconstructed, replayed, and policy-level claims.
\item A sealed execution manifest that binds code, configuration, roles, selectors, adapters, schemas, and release identity.
\item A seven-layer evidence graph with separately retained inputs, registered derivations, and a shared correlation identity.
\item A proof that hash equality alone cannot establish independent derivation.
\item A claim-promotion predicate that is conjunctive rather than compensating.
\item Tamper, drift, omission, duplication, and forbidden-authority controls.
\item A reproducible retained-package specification with allowlists, manifests, SBOMs, clean runs, and independent review.
\item A historical provenance model that distinguishes original release identities from later reproducibility attestations.
\item An exactly-once recovery evidence object that binds worker failure and projection rebuild to an already committed operation without a duplicate financial effect.
\item A twelve-parent FIA case study, including mandatory FIA-011A/B subcases, and a strict record of what remains not proven.
\end{enumerate}

\section{Problem statement and threat model}
\subsection{What an E2E claim means}
Let $q$ be a protocol claim, such as ``a partial settlement reduced debt once and appeared consistently in the user and audit interfaces.'' An end-to-end claim is not a single Boolean test. It is a statement about a chain of related observations under a specific system identity.

Let
\[
\mathcal L=\{L_1,L_2,\ldots,L_7\}
\]
be evidence layers and let $S$ be a sealed execution manifest. A claim $q$ is supported only if each required layer has a registered derivation from retained inputs, all layers bind to the same manifest and correlation identity, the required derivation-root topology is satisfied, specified invalid mutations are rejected (so the corresponding negative-control tests pass), and the complete package is reproducible and independently reviewable.

\subsection{Adversarial evidence failures}
The threat model includes both malicious and accidental failures:
\begin{itemize}
\item a result is detached from the commit or deployment that produced it;
\item an API points to a stale address or selector;
\item fixture expectations are copied into actual-state fields;
\item multiple projections are generated from one sealed artifact and presented as independent runtime observations;
\item evidence records or layers from distinct runs, manifests, or correlations are spliced into one synthetic cohort claim;
\item event ordering is inferred from file order rather than chain coordinates;
\item a journal entry is omitted, duplicated, or reordered;
\item a browser displays cached or reference data;
\item a package includes unallowlisted files, secrets, logs, or a self-referential checksum;
\item a clean-run comparison excludes semantically relevant differences;
\item a worker restarts after chain confirmation and financially re-executes the command;
\item a historical phase is assigned a retroactive tag or expected hash that never existed;
\item a green bounded test is promoted into a broader claim.
\end{itemize}

The goal is not to make all evidence impossible to forge in an absolute sense. The goal is to make the claim boundary explicit, machine-checkable, and difficult to broaden without producing new evidence.

\section{Typed evidence and claim taxonomy}
\begin{definition}[Evidence item]
An evidence item is a tuple
\[
E=(\tau,s,d,h,c,m),
\]
where $\tau$ is the evidence type, $s$ the retained source identity, $d$ the derivation description, $h$ the digest, $c$ the correlation identity, and $m$ the manifest identity.
\end{definition}

The primary evidence types are:
\begin{itemize}
\item \texttt{FORMAL}: theorem under stated assumptions;
\item \texttt{DETERMINISTIC\_BENCHMARK}: exact registered cases;
\item \texttt{OBSERVED\_LOCAL}: chain receipt, state, or process fact observed locally;
\item \texttt{RECONSTRUCTED\_PROJECTION}: materialized service/API/browser/audit state;
\item \texttt{MECHANICAL\_REPLAY}: result of a frozen workflow replay;
\item \texttt{POLICY}: owner or governance design choice;
\item \texttt{INCOMPLETE}: a required gate inside the registered scope lacks evidence;
\item \texttt{NOT\_EVALUATED}: a broader claim lies outside that tested scope.
\end{itemize}

\begin{proposition}[Non-substitutability of evidence types]
No number of reconstructed projections can substitute for an absent observed canonical transaction or state transition, and no number of unit tests can substitute for an absent historical replay when the claim includes replayability.
\end{proposition}
\begin{proof}
The evidence types assert different predicates. A projection asserts consistency with its source, not occurrence of a canonical financial transition. A unit test asserts a property of the test environment, not reproducibility of a historical workflow. Summing evidence of one type cannot establish a logically distinct missing predicate.
\end{proof}

\section{Sealed execution manifest}
The manifest anchors what system was exercised. It should contain:
\begin{enumerate}[label=(\roman*)]
\item source commit or tag and repository identity;
\item chain ID and execution mode;
\item contract addresses, creation/deployed bytecode, ABIs, and storage layouts;
\item roles, owners, selector policy, and pause target set;
\item adapter and venue-emulator identity;
\item OpenAPI, schema, generated-client, and UI build identities;
\item dependency lock and toolchain identities;
\item a digest stored outside the manifest body that it hashes.
\end{enumerate}

Let $S$ be the RFC 8785 canonical JSON representation of the manifest body, excluding its own digest and downstream evidence/package digests \citep{rfc8785}. The receipt wrapper is formed only after hashing. Its digest is
\[
h_M=\Hh(S).
\]

\begin{definition}[Manifest-bound evidence]
An evidence item is manifest-bound if its canonical representation includes $h_M$ and the verifier rejects any mismatch between $h_M$ and the runtime identities used to produce the item.
\end{definition}

\begin{proposition}[Drift visibility]
If address, bytecode, ABI, role, selector, adapter, schema, or client identity is included in $S$, any change to those fields changes $h_M$ under collision resistance. A result under the changed system cannot retain the old manifest identity without a hash collision or a verifier defect.
\end{proposition}

The manifest is analogous to supply-chain provenance: it binds the artifact to the source and process that produced it \citep{torresarias2019intoto,slsa2023,nistssdf2022}. Unlike a generic build attestation, it also binds financial authority and permitted command surfaces.

\section{Seven retained evidence layers and registered derivations}
The evidence model uses the layers in \cref{tab:layers}, and their manifest/correlation linkage is illustrated in \cref{fig:layers}. The word \emph{retained} is deliberate: a layer can be useful without being an independent root.

\begin{table}[t]
\centering\small
\caption{Seven retained evidence layers.}
\label{tab:layers}
\begin{tabularx}{\textwidth}{@{}P{0.21\textwidth}P{0.34\textwidth}Y@{}}
\toprule
Layer & Source & Required property \\
\midrule
1. Simulator expectation & registered vector/fixture & expectation only; never financial authority \\
2. Canonical contract state & manifest addresses, ABI views, ordered events & observed from the deployed graph \\
3. Append-only journal & transaction/event identifiers & complete, ordered, no silent overwrite \\
4. Service projection & journal replay & registered materialization, non-authoritative \\
5. Authenticated API & service projection under role/session & registered derivation, route, identity, response, no bypass \\
6. Browser projection & API/service state after a real role control or typed browser client & visible state tied to command/correlation; not presumed independent \\
7. Audit reconstruction & correlation and journal replay & reviewer-readable, separately rebuilt where policy requires \\
\bottomrule
\end{tabularx}
\end{table}

The labels L1--L7 in this paper are explanatory labels corresponding to the table; they do not renumber an immutable dataset schema. For layer $i$, let $I_i$ be the retained input set and $D_i$ the derivation of its output body (excluding its own digest and downstream commitments). Then
\[
h_i=\Hh(\mathrm{JCS}(D_i(I_i,h_M,c))).
\]

\begin{definition}[Registered derivation]
A retained layer record has a registered derivation if it identifies its input set, transformation, root ancestry, implementation/runtime identity, canonical output-body projection, output digest, and correlation/manifest bindings. Body, canonical bytes, digest, and descriptive receipt are different objects. Heterogeneous artifact digests need not be equal; a registered semantic projection supplies the common comparison. Registration does not by itself imply independence.
\end{definition}

\begin{definition}[Multi-root corroboration]
For claim $q$, the previously registered failure model and policy determine the required source-origin, failure-domain, and processing-path relations among its retained observations $R(q)$. Corroboration passes only when every required relation is evidenced; copies or relabelings of one observation do not create another root, and unknown provenance does not establish independence. A graph-cut condition alone is not equivalent to this property. Distinct acquisition origins may share an aggregator, while distinct files may share one origin. The conclusion is claim-specific corroboration under the declared model, not statistical independence.
\end{definition}

\begin{theorem}[Copied-hash insufficiency]
Suppose $h_1=\cdots=h_7=h$, but every layer was generated by copying one fixture value $h$. Then equality provides no evidence that any actual layer observed or reconstructed the claimed state.
\end{theorem}
\begin{proof}
The equality is entailed by the copying procedure alone and is independent of the system execution. Therefore it has zero discriminating power between a correct execution, an incorrect execution, and no execution. Separately retained inputs with registered derivations are necessary to connect equality to the protocol; claims of independent runtime participation would additionally require a proposed multi-root corroboration gate. That gate was not evaluated for the registered local cohort and remains \texttt{NOT\_EVALUATED}.
\end{proof}

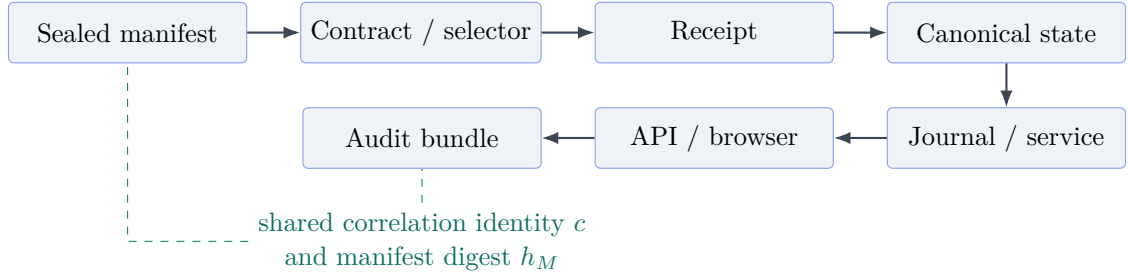
\begin{figure}[t]
\centering
\begin{tikzpicture}[node distance=0.58cm and 0.7cm,
  layer/.style={draw=axblue!65,fill=axgray,rounded corners=2pt,minimum width=3.15cm,minimum height=0.8cm,align=center,font=\small},
  arr/.style={-{Latex[length=2.2mm]},thick,draw=axnavy!80}]
\node[layer] (m) {Sealed manifest};
\node[layer,right=of m] (c) {Contract / selector};
\node[layer,right=of c] (r) {Receipt};
\node[layer,right=of r] (s) {Canonical state};
\node[layer,below=of s] (j) {Journal / service};
\node[layer,left=of j] (a) {API / browser};
\node[layer,left=of a] (u) {Audit bundle};
\draw[arr](m)--(c);\draw[arr](c)--(r);\draw[arr](r)--(s);\draw[arr](s)--(j);\draw[arr](j)--(a);\draw[arr](a)--(u);
\node[below=0.45cm of u,align=center,font=\small,text=axteal!60!black] (corr) {shared correlation identity $c$\\and manifest digest $h_M$};
\draw[dashed,draw=axteal!70!black] (corr.west)-|(m.south);\draw[dashed,draw=axteal!70!black] (corr.north)--(u.south);
\end{tikzpicture}
\caption{Manifest-bound retained evidence chain. The diagram is sequential for readability; verification also checks cross-links, provenance, and registered derivations. A proposed multi-root corroboration gate would be required for a stronger independence claim, but was not evaluated here.}
\label{fig:layers}
\end{figure}

\section{Correlation identity}
A scenario must be traceable across user, service, chain, and audit layers. Define
\[
c=\Hh(\mathrm{JCS}(B_c)),
\]
Here $B_c$ is the fixed JSON input body binding scenario, manifest digest, role, command, and idempotency domain under the retained serialization profile. The tuple description does not specify an additional wire format. Exact historical digests depend on that profile and its fixed vectors; the notation does not assert new codec compatibility. The retained correlation ledger should include:
\begin{itemize}
\item UI route, control ID, role/session ID;
\item API endpoint, command ID, payload and idempotency-key digests;
\item saga, inbox, and outbox IDs;
\item emulator order/match/settlement IDs;
\item settlement-evidence ID;
\item contract address, selector, calldata digest, transaction hash;
\item ordered event identifiers;
\item journal entries and projection identities;
\item browser and audit reconstruction identities.
\end{itemize}

\begin{proposition}[Correlation exclusion]
If the verifier rejects events, addresses, selectors, or projections that are not manifest-authorized or not bound to $c$, a legacy bridge or off-path action cannot silently satisfy the canonical claim.
\end{proposition}

A correlation ledger is not itself financial authority. Its purpose is to make the exact authority path inspectable and to fail when a required observation is missing.

\begin{definition}[Cross-cohort splice]
A cross-cohort splice is a constructed evidence package in which at least one load-bearing record or layer is drawn from a different run, manifest, scenario correlation, or parent-record identity while the package is presented as one coherent cohort.
\end{definition}

\begin{definition}[No-splice gate]
For claim $q$, the no-splice gate passes only if every required record and derivation edge resolves to the same declared cohort identity, each scenario correlation occurs exactly once, and no layer is admitted through an identity not reachable from that cohort's manifest and ledger root.
\end{definition}

\begin{proposition}[No-splice soundness under complete identity validation]
Assume collision resistance, a trusted fixed cohort commitment, and complete validation of run, manifest, correlation, scenario-record, and layer-root identities against that commitment. Then replacing a required layer with an object from another cohort causes the no-splice gate to fail unless a checked identity is also forged.
\end{proposition}
\begin{proof}
A foreign object differs in at least one declared run, manifest, correlation, scenario record, or root ancestry. Complete validation checks these fields against the cohort registry. The substitution therefore changes a checked digest or violates an identity predicate.
\end{proof}

\section{Evidence graph and independence model}
\subsection{Directed evidence graph}
As illustrated in \cref{fig:derivation-roots}, represent an evidence package as a directed acyclic graph
\[
G_E=(V_E,E_E),
\]
where each vertex is an evidence item and each directed edge records a derivation or reference relation. Root vertices are retained observations or registered expectations; non-root vertices are projections, normalizations, package manifests, or review decisions.

A valid graph must satisfy:
\begin{enumerate}[label=(E\arabic*)]
\item every non-root vertex names all load-bearing parents;
\item every edge is typed, for example \textsf{observed-from}, \textsf{derived-from}, \textsf{replayed-from}, or \textsf{reviewed-against};
\item each scenario-scoped vertex binds to that scenario's correlation and approved manifest; collection, package, and replay records bind multiple scenarios through validated membership edges;
\item the computational digest/derivation graph is acyclic; an excluded descriptive back-reference is not an edge in that graph and does not permit checksum self-reference;
\item a layer claimed to have an independent runtime root cannot have the compared layer's normalized output as its only parent.
\end{enumerate}

\begin{definition}[Derivation cut]
A derivation cut for claim $q$ is a set of supporting observations or intermediate artifacts whose removal disconnects all retained support for $q$. It describes dependency topology, not independently established failure domains.
\end{definition}

A single intermediate artifact can be a support bottleneck without erasing distinct upstream acquisition origins. Conversely, separate graph roots can conceal a shared acquisition or failure domain. Thus cuts diagnose dependencies but do not replace the registered independence assumptions. Domain-separated hashes do not create new observations.

\begin{proposition}[Common-root limitation]
If every actual-layer record is a deterministic function of one generated artifact $z$ and no independent observation of the claimed financial state is retained, then the evidence proves consistency with $z$, not independent participation of the corresponding runtime layers.
\end{proposition}
\begin{proof}
All evidence paths factor through $z$. Replacing the real system with any generator that emits the same $z$ leaves every derived record unchanged. Therefore the package cannot discriminate real participation from generated consistency.
\end{proof}

\subsection{Independence classes}
The framework distinguishes the following dimensions of independence and separation:
\begin{itemize}
\item \textbf{syntactic independence}: different files, schemas, or domain tags;
\item \textbf{implementation independence}: different code paths or languages;
\item \textbf{runtime independence}: separately running components observe or reconstruct the event;
\item \textbf{verifier-process separation}: a distinct process validates the retained package;
\item \textbf{operator independence}: a separately identified human or organization performs or controls the review under a declared independence policy.
\end{itemize}
None implies the next. In particular, two processes under one operator establish process separation, not operator independence. The final claim matrix should state which class is achieved for each layer.

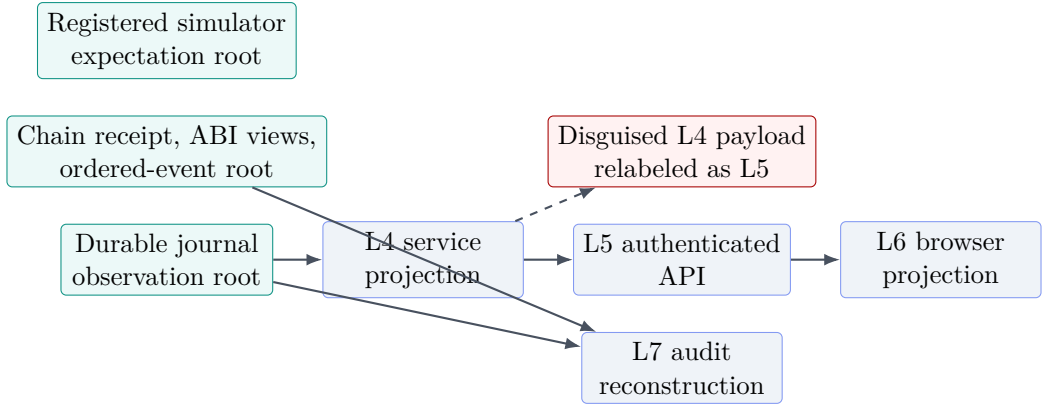
\begin{figure}[t]
\centering
\begin{tikzpicture}[node distance=0.48cm and 0.65cm,
  root/.style={draw=axteal!80!black,fill=axteal!8,rounded corners=2pt,minimum width=2.8cm,minimum height=0.72cm,align=center,font=\small},
  layer/.style={draw=axblue!65,fill=axgray,rounded corners=2pt,minimum width=2.65cm,minimum height=0.72cm,align=center,font=\small},
  bad/.style={draw=red!65!black,fill=red!5,rounded corners=2pt,minimum width=2.9cm,minimum height=0.72cm,align=center,font=\small},
  arr/.style={-{Latex[length=2.1mm]},thick,draw=axnavy!75}]
\node[root] (sim) {Registered simulator\\expectation root};
\node[root,below=of sim] (chain) {Chain receipt, ABI views,\\ordered-event root};
\node[root,below=of chain] (journal) {Durable journal\\observation root};
\node[layer,right=of journal] (service) {L4 service\\projection};
\node[layer,right=of service] (api) {L5 authenticated\\API};
\node[layer,right=of api] (browser) {L6 browser\\projection};
\node[layer,below=of api] (audit) {L7 audit\\reconstruction};
\node[bad,above=of api] (fake) {Disguised L4 payload\\relabeled as L5};
\draw[arr](journal)--(service);\draw[arr](service)--(api);\draw[arr](api)--(browser);\draw[arr](journal)--(audit);\draw[arr](chain)--(audit);\draw[arr,dashed](service)--(fake);
\end{tikzpicture}
\caption{Registered derivation-root topology. Service, API, and browser layers may share one projection root. The framework therefore reports root ancestry and rejects an L4 object merely repackaged as L5. Multi-root corroboration is imposed only for claims that require independent participation.}
\label{fig:derivation-roots}
\end{figure}

\begin{proposition}[Disguised-layer rejection]
If an L4 service object is serialized under an L5 label without an authenticated API observation and distinct registered derivation, the evidence graph preserves the L4 root ancestry and the L5 gate fails. This distinct-root negative control is specified but was not evaluated in the registered local cohort.
\end{proposition}

\subsection{Evidence completeness and inclusion-minimal core}
An evidence package should be complete enough to reproduce the declared claim while avoiding unnecessary secrets, logs, and personal data. Let $F_{\mathrm{core}}$ be the evidence core and let $G_q=\{g_1,\ldots,g_n\}$ be the fixed scientific gate set for claim $q$. Fix a verification context $A_q$ containing the required manifests, schemas, commitments, policy, and declared external toolchain. Separately classified licences, readme files, and transport metadata form $F_{\mathrm{anc}}$; any ancillary record required by the verification context remains part of $A_q$ and is not thereby dispensable.

For $S\subseteq F_{\mathrm{core}}$, write $R_q(S;A_q)$ for the vector of registered gate results, with an explicit undefined value when a required gate cannot be evaluated. The core is \emph{claim-complete} when every required gate is defined and $R_q(F_{\mathrm{core}};A_q)$ equals the registered result vector. It is \emph{inclusion-minimal relative to $G_q$ and $A_q$} only if no proper subset $S\subsetneq F_{\mathrm{core}}$ both leaves every required gate defined and preserves that same registered result vector.

This criterion concerns set inclusion, not minimum cardinality, minimum byte size, or an asserted minimization algorithm. A single-file deletion check establishes only single-deletion irredundancy unless additional monotonicity assumptions are proved. Arbitrary caches, videos, process dumps, and private directories belong in neither class unless a registered gate or fixed verification context explicitly requires them.

\section{Predicate classes and claim promotion}
Not every promotion predicate is monotone under appended evidence. The proposed framework classifies each gate as one of:
\begin{itemize}
\item \texttt{MONOTONE\_POSITIVE}: once true under append-only valid evidence, it remains true;
\item \texttt{TERMINAL\_CLOSURE}: evaluated only when the run or package is sealed;
\item \texttt{ABSENCE\_NEGATIVE\_CONTROL}: absence of a prohibited event over a bounded terminal domain;
\item \texttt{RECOMPUTED\_COMPARISON}: a clean-run, replay, or package comparison recomputed from current inputs;
\item \texttt{MULTI\_ROOT\_CORROBORATION}: a proposed derivation-topology predicate recomputed when roots or edges change; it was not evaluated for the registered local cohort.
\end{itemize}
Blanket upward monotonicity is invalid for the last four classes. Package closure can be broken by an extra member, an absence gate can fail when a prohibited event arrives, and a comparison or root topology can change when evidence is appended.

\section{Claim promotion as a conjunctive predicate}
Fix a nonempty, prospectively registered gate set for claim $q$ and its stated scope. Let the required gates be $g_1(q),\ldots,g_n(q)\in\{0,1\}$, where $1$ means positively established. An unknown required gate cannot be assigned $1$. Define
\[
P(q)=\bigwedge_{k=1}^n g_k(q).
\]
The claim is promoted only if $P(q)=1$.

For a full local E2E claim, representative gates include:
\begin{enumerate}[label=(G\arabic*)]
\item sealed manifest and approved command path;
\item canonical receipt or typed no-mutation rejection;
\item canonical state observation;
\item seven retained derivations;
\item shared correlation identity;
\item no-splice cohort-identity gate;
\item mandatory negative controls;
\item two clean complete packages;
\item the registered selected historical replay;
\item independent final-package review;
\item no forbidden bridge, external write, production asset, or leverage above the scoped limit.
\end{enumerate}

\begin{theorem}[No compensating promotion]
Under the conjunctive predicate, passing additional non-critical tests cannot compensate for a failed required safety or provenance gate.
\end{theorem}
\begin{proof}
A conjunction is false when any conjunct is false, independently of the number of true conjuncts. Therefore performance, code coverage, or extra passing cases cannot hide a missing required layer.
\end{proof}

This rule differs from a weighted maturity score. A score may be useful for project management, but it is unsuitable for a claim whose meaning includes every specified layer. The resulting promotion sequence is shown in \cref{fig:ladder}.

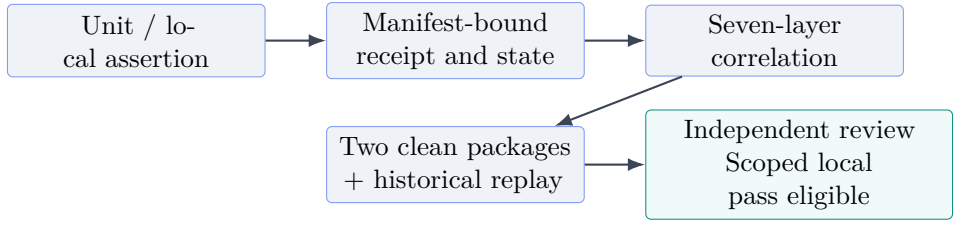
\begin{figure}[t]
\centering
\begin{tikzpicture}[node distance=0.62cm and 0.8cm,
  box/.style={draw=axblue!65,fill=axgray,rounded corners=2pt,text width=3.15cm,minimum height=0.78cm,align=center,font=\small},
  gate/.style={draw=axteal!80!black,fill=axteal!7,rounded corners=2pt,text width=3.8cm,minimum height=0.78cm,align=center,font=\small},
  arr/.style={-{Latex[length=2.3mm]},thick,draw=axnavy!80}]
\node[box] (u) {Unit / local assertion};
\node[box,right=of u] (r) {Manifest-bound\\receipt and state};
\node[box,right=of r] (l) {Seven-layer\\correlation};
\node[box,below=of r] (p) {Two clean packages\\+ historical replay};
\node[gate,right=of p] (f) {Independent review\\Scoped local pass eligible};
\draw[arr](u)--(r);\draw[arr](r)--(l);\draw[arr](l)--(p);\draw[arr](p)--(f);
\end{tikzpicture}
\caption{Claim-promotion ladder. Every stage adds a qualitatively distinct predicate; later language must not be applied early.}
\label{fig:ladder}
\end{figure}

\section{Tamper and drift controls}
\subsection{Mutation families}
The verifier should intentionally mutate:
\begin{itemize}
\item manifest, ABI, bytecode, storage, address, role, or selector identity;
\item receipt status, log digest, event order, or transaction index;
\item fixture hash propagated into an actual layer;
\item omitted, duplicated, or reordered journal entries;
\item stale service/API/browser projections;
\item an L4 service payload disguised as an L5 authenticated-API observation (a specified but unevaluated distinct-root control);
\item a scenario record or layer spliced from a different run, manifest, correlation, or cohort root;
\item correlation IDs and idempotency domains;
\item output-root, ZIP members, checksum files, and allowlists;
\item external-write markers, non-loopback routes, production identifiers, or over-$2\times$ settings;
\item duplicate and conflicting settlement evidence;
\item pause target failure, premature unpause, or recovery reexecution.
\end{itemize}

\begin{assumption}[Fixed commitment and complete validation]
The hash function is collision resistant, the verifier checks every declared load-bearing field and binding, and the expected source/package commitment is fixed independently of the candidate being checked. An attacker cannot replace that trusted expected commitment together with the candidate.
\end{assumption}

\begin{theorem}[Specified-tamper detection relative to a commitment]
Under the fixed-commitment assumption, a mutation that changes a committed canonical input or violates a required schema/policy predicate is rejected, except with the negligible collision probability. Recomputing a changed candidate's own checksum is not sufficient to match the fixed expected commitment.
\end{theorem}
\begin{proof}
A changed committed input has different canonical bytes and hence a different checked digest under collision resistance. A mutation outside that class is rejected when it violates a checked predicate. The attacker cannot repair either outcome merely by replacing a self-reported checksum because validation is anchored to the independent expected commitment.
\end{proof}

The result does not reject harmless differences removed by the declared canonicalization, prove the truth of the committed source, or detect omissions outside the schema. Without a trusted anchor, an internally consistent replacement of all files and their checksums is a direct counterexample to unconditional tamper detection. Mutation testing can expose missing checks \citep{jia2011mutation}; it does not create authenticity by itself.

\subsection{Differential and property-based validation}
Independent implementations should canonicalize and verify the same objects. Differential testing catches implementation-specific mistakes \citep{mckeeman1998differential}; property-based testing exercises invariants across generated inputs \citep{claessen2000quickcheck}; smart-contract fuzzers and multi-tool comparisons provide complementary coverage \citep{wustholz2020harvey,durieux2020smartbugs}. None of these methods alone proves a claim, but together they strengthen the evidence package.

\section{Retained evidence package}
A portable evidence package must be reviewable without exposing secrets or relying on temporary paths. The allowlisted contents may include:
\begin{itemize}
\item manifest, schemas, source/release identity, and toolchain metadata;
\item per-scenario command, receipt, state, event, journal, service, API, browser, audit, reconciliation, and validator records;
\item negative-control results;
\item two clean-run indexes and comparison;
\item replay attestation;
\item SBOM and dependency/license inventory;
\item privacy/secret scan and independent review report.
\end{itemize}
Excluded by default are credentials, keys, cookies, raw logs, videos, dumps, private RPCs, personal data, raw third-party venue data, historical harness evidence presented as canonical, and self-referential checksum entries.

\begin{definition}[Package identity]
For an ordered allowlisted payload member set $F=(f_1,\ldots,f_k)$, the package identity is the digest of the deterministic archive containing $F$ and its member manifest. The member manifest lists payload digests, sizes, modes, and paths but excludes its own checksum entry. The outer archive digest is stored outside the archive. These exclusions distinguish archive bytes from semantic case identity without changing historical package bytes.
\end{definition}

Two clean packages may be required to be byte-identical. When raw runtime fields are inherently variable, the policy must separate raw identities from a normalized semantic identity and justify each excluded field. A normalized comparison cannot be silently described as byte-identical.

\section{Historical replay provenance}
Long programmes often lack a frozen replay wrapper or aggregate output hash for an early phase. Backdating a tag or assigning a historical expected hash creates fictitious evidence.

\begin{definition}[Post-hoc reproducibility attestation]
A post-hoc reproducibility attestation binds an already recorded source boundary, a later-created replay wrapper, the pinned environment reconstructed from historical lockfiles, and the newly observed replay result. It is not the original historical release identity.
\end{definition}

For phase $p$, record
\[
B_p=(\textsf{boundary type},\textsf{commit/tag},\textsf{historical replay status},\textsf{new attestation class}).
\]
If no historical aggregate hash existed, the field remains null with a reason. This discipline preserves provenance while still permitting useful later replay.

\begin{proposition}[No retroactive identity]
A post-hoc replay hash cannot be represented as an original historical expected-output hash without making a false statement about what was frozen at the earlier time.
\end{proposition}

\section{Exactly-once recovery evidence}
A financial recovery claim must prove more than application availability. For command $x$, retain:
\begin{enumerate}[label=(\arabic*)]
\item chain-confirmed transaction and receipt;
\item durable checkpoint written before process termination;
\item old process identity and termination evidence;
\item fresh process identity;
\item inbox/outbox replay record;
\item registry replay/no-op classification;
\item one committed financial effect in the stable operation domain; record the narrower one-transaction criterion when that is the tested fixture;
\item rebuilt journal, service, API, browser, and audit projections.
\end{enumerate}

\begin{definition}[Exactly-once evidence object]
An exactly-once evidence object is the tuple
\[
E_{\mathrm{xo}}=(x,h_M,c,\tau,C,p_{old},p_{new},R,\Delta_f,\Pi),
\]
where $\tau$ is the canonical transaction, $C$ the durable checkpoint, $R$ the replay result, $\Delta_f$ the duplicate financial delta, and $\Pi$ the rebuilt projection set.
\end{definition}

For the retained FIA-010 post-checkpoint fixture, the observed acceptance condition is
\[
\#\{\text{canonical financial transactions for }x\}=1,\qquad \Delta_f=0.
\]
The general safety property counts committed effects per stable operation/key/payload, not transaction submissions. Reverted attempts and distinct settlement/release steps need not share one transaction identifier. Admission alone is not application; a checkpoint-only trace does not test every earlier crash window. Exactly-once message delivery and unconditional eventual completion are not claimed.

\section{FIA case study}
The twelve FIA cases cover the evidence questions summarized in \cref{tab:fia-evidence}:
\begin{longtable}{@{}P{0.09\textwidth}P{0.30\textwidth}P{0.53\textwidth}@{}}
\caption{FIA evidence questions.}\label{tab:fia-evidence}\\
\toprule
FIA & Financial scenario & Evidence question \\
\midrule
\endfirsthead
\toprule
FIA & Financial scenario & Evidence question \\
\midrule
\endhead
01 & successful admission & can UI/API/origination receipts and canonical state be correlated? \\
02 & exact $2\times$ admission & is the leverage policy and calldata identity preserved? \\
03 & risk rejection & is typed rejection proven with zero mutation? \\
04 & partial then final settlement & do distinct admitted receipts bind bounded partial application and final settlement in one lifecycle record? \\
05 & residual and collateral release & are trader residual and collateral released once after canonical debt-clearance eligibility? \\
06 & liquidation to hard-flat & are authorized liquidation request/execution and the hard-flat transition bound to one correlation? \\
07 & waterfall & do reserve, LBP, and Senior state changes conserve ordered loss? \\
08 & conflict & is quarantine observed with zero debt mutation? \\
09 & Senior withdrawal queue & do queued \texttt{CreditPool} shares remain loss-participating until service? \\
10 & recovery & is one financial transaction preserved across process restart? \\
11 & conjunctive capability transition & do FIA-011A stale-before-admission and FIA-011B book-epoch change with $D_t>0$ each pass retained-evidence, causal, fail-closed/no-mutation, negative-control, replay, and report gates? \\
12 & pause/recovery & are target set, guardian pause, timelock, and unpause ordered correctly? \\
\bottomrule
\end{longtable}

The cases are not an aggregate score. A missing evidence layer in FIA-10 cannot be offset by extra tests in FIA-01. FIA-011 is one parent assertion and passes only when both FIA-011A and FIA-011B close with retained raw and normalized evidence, manifest-bound identity, exact transition provenance, causal checks, expected no-mutation or fail-closed semantics, negative controls, replay, and explicit per-subcase report coverage. One passing subcase cannot close the parent, which remains blocked while either evidence path is incomplete. The same principle applies to separately registered benchmark extensions, which are not part of this cohort.

\section{Registered local evidence result}
\interimnote

The registered local cohort contains twelve parent FIA assertions. Each parent retains seven \texttt{ACTUAL\_EVIDENCE} positions, for a total of 84. The bounded result is reported through aggregate identities and counts; confidential operational identifiers remain outside the distributed evidence boundary.

The registered local evidence milestones are summarized in \cref{tab:promotion}.

\begin{table}[t]
\centering\small
\caption{Registered local evidence-chain milestones.}
\label{tab:promotion}
\begin{tabularx}{\textwidth}{@{}P{0.35\textwidth}Y@{}}
\toprule
Object & Observed identity/result \\
\midrule
All-12 cohort & 12 parent records, 84 retained positions \\
Clean materializations & two clean materializations \\
Shared file index & one byte-identical canonical file index \\
Comparison & byte-identical at registered boundary \\
Strict archive & separate verifier accepted the registered archive \\
Ordered replay & selected replay bound to registered archive \\
\bottomrule
\end{tabularx}
\end{table}

The two clean trees were produced through separate source reads and have one canonical file-index identity. The strict packager produced fixed archive content, and a separate verifier reopened the archive bytes. The selected ordered replay preserves registered causal ordering before binding to the strict archive.

FIA-011 provides a useful evidence-contract example. Its parent record contains two mandatory retained subcases: stale capability before admission and a later book-epoch change after a position has opened. The source-owned adapter and focused tests require both members; FIA-03 or a one-receipt placeholder cannot satisfy the parent.

\subsection{Observed, interpreted, and future claims}
\textbf{Observed:} the identities, counts, byte comparison, verifier acceptance, and replay binding above. \textbf{Interpretation:} the evidence architecture preserves a chain of custody from financial assertion through projections, deterministic materialization, packaging, and replay. \textbf{Future hypotheses:} external venues, public networks, real custody, terminal settlement, and production operations require new evidence contracts rather than reinterpretation of this cohort.

\section{Evidence package schemas}
\subsection{Scenario evidence record}
A per-scenario evidence record should contain the following groups:
\begin{enumerate}[label=(\alph*)]
\item \textbf{identity}: schema, scenario, manifest, correlation, source release;
\item \textbf{command}: role, session class, route, endpoint, command, idempotency domain, payload digest;
\item \textbf{chain}: target, selector, calldata digest, transaction, receipt, ordered events;
\item \textbf{state}: terminal block, ABI-view inputs, canonical state projection and digest;
\item \textbf{application}: journal, service, API, browser, and audit inputs and digests;
\item \textbf{negative controls}: mutation identifiers, expected rejection, observed rejection;
\item \textbf{recovery}: checkpoint/process/replay fields when applicable;
\item \textbf{verdict}: gate vector, missing layers, classification, reviewer notes.
\end{enumerate}

A simplified record is:
\begin{verbatim}
{
  "schema": "axient/e2e-evidence/v1",
  "scenario_id": "FIA-010",
  "manifest_hash": "...",
  "correlation_id": "...",
  "command": {...},
  "chain": {...},
  "canonical_state": {...},
  "journal": {...},
  "service": {...},
  "api": {...},
  "browser": {...},
  "audit": {...},
  "recovery": {...},
  "negative_controls": [...],
  "verdict": {...}
}
\end{verbatim}

\subsection{Release index}
The release index should record every scenario, file, digest, test, replay, and review artifact. It also records deliberate absences. A field such as \texttt{historical\_expected\_output\_hash=null} is preferable to inventing a value; the reason field explains that the hash was not frozen at the historical phase.

\subsection{Allowlist and path safety}
The package builder must reject absolute paths, parent-directory traversal, symlinks where prohibited, duplicate ZIP members, executable binaries not declared by policy, and files outside the allowlist. Checksum files exclude their own line to avoid circular identity. Build outputs should be generated in clean OS-temporary or designated durable evidence roots, never collected by recursively archiving an active development directory.

\section{Independent review protocol}
An independent review has two stages.

\subsection{Mechanical verification}
The reviewer verifies:
\begin{enumerate}[label=(M\arabic*)]
\item outer and inner archive hashes;
\item member paths, sizes, modes, and manifest entries;
\item schemas and canonicalization;
\item source, build, toolchain, and manifest identities;
\item every gate and negative control;
\item two clean-run comparison;
\item historical replay provenance;
\item secret, privacy, data-rights, and license boundaries.
\end{enumerate}

\subsection{Semantic verification}
The reviewer then asks whether the retained artifacts actually support the language used. Examples:
\begin{itemize}
\item Does a chain receipt represent the intended selector or merely a harness event?
\item Was the browser state caused by a role control or copied from a fixture?
\item Does a recovery report show a fresh process without a duplicate committed financial effect?
\item Do five projections have separately observed runtime roots under the registered derivation graph, or are they five views from one generator?
\item Does the declared FULL\_E2E scope include the relevant contracts and services?
\end{itemize}

\begin{definition}[Evidence-review outcome]
The review outcome is one of the following:\par
\begin{center}
\texttt{PASS}\qquad
\texttt{PASS\_WITH\_LIMITATIONS}\qquad
\texttt{FAIL\_CLOSED}.
\end{center}
Each outcome is accompanied by a machine-readable finding registry. The outcome does not change the evidence; it constrains the claims that may be made from it.
\end{definition}

\section{Detailed tamper matrix}
The principal mutation targets and required detections are listed in \cref{tab:tamper}.
\begin{longtable}{@{}P{0.23\textwidth}P{0.30\textwidth}P{0.38\textwidth}@{}}
\caption{Representative tamper controls.}\label{tab:tamper}\\
\toprule
Target & Mutation & Required detection \\
\midrule
\endfirsthead
\toprule
Target & Mutation & Required detection \\
\midrule
\endhead
Manifest & change address, ABI, bytecode, role, selector, adapter, schema & manifest verifier rejects before execution or materialization \\
Receipt & flip status, tx digest, topic, data, chain ID & receipt/state verifier rejects \\
Event order & swap block/transaction/log ordering & canonical state digest changes or ordering gate rejects \\
State & copy fixture expected hash into actual state & source-provenance gate rejects \\
Journal & omit or duplicate canonical event reference & journal completeness/replay differs \\
Service & materialize stale or off-manifest state & service projection/manifest gate rejects \\
API & alter role/session or response field & authenticated response and hash mismatch \\
Browser & show cached/reference result & browser/source/correlation proof fails \\
Audit & omit command or selector edge & reconstruction completeness fails \\
Cohort & splice record or layer from another run/manifest/correlation & no-splice identity gate fails \\
Package & inject member, change path, remove file & allowlist/manifest/checksum fails \\
Recovery & repeat a committed financial effect & exactly-once gate fails \\
Boundary & include remote RPC, external write, real asset, over-limit leverage & environment policy fails \\
\bottomrule
\end{longtable}
\FloatBarrier

\section{Claim-language governance}
Technical evidence should map to approved language. Representative mappings are given in \cref{tab:language}.
\begin{table}[t]
\centering\small
\caption{Evidence-to-language examples.}
\label{tab:language}
\begin{tabularx}{\textwidth}{@{}P{0.31\textwidth}P{0.31\textwidth}Y@{}}
\toprule
Evidence & Allowed language & Prohibited promotion \\
\midrule
unit/local test & ``the local test passed'' & ``the protocol is E2E validated'' \\
manifest-bound receipt & ``a canonical local receipt was observed'' & ``the full user workflow completed'' \\
seven retained layers & ``cross-layer evidence with registered derivations was retained'' & ``independent systems all proved the result'' unless a separately evaluated multi-root gate is shown \\
two clean packages & ``the package rebuilt reproducibly under the stated policy'' & ``production is deterministic'' \\
full replay + independent review & ``the scoped local E2E claim passed'' & ``production-safe, audited, or live-venue validated'' \\
\bottomrule
\end{tabularx}
\end{table}

This governance is not cosmetic. It prevents a technically correct narrow result from becoming a misleading public claim.

\appendix
\section{Formal properties of evidence promotion}
\subsection{Predicate-class-aware monotonicity}
If the required gate set changes from $G$ to $G'\supseteq G$, a claim that passed under $G$ may fail under $G'$. This is not evidence regression; it is a stricter claim definition. Historical results retain their original scope and are not retroactively relabeled.

\begin{proposition}[Monotonicity of positive append-only gates]
Fix a manifest and scope. Let $g$ be classified \texttt{MONOTONE\_POSITIVE}, and suppose its truth depends only on the presence and validity of an append-only witness set. Appending additional valid, non-contradictory evidence cannot change $g=1$ to $g=0$.
\end{proposition}
\begin{proof}
By definition, the witness establishing $g=1$ remains present and valid after append-only extension. Because the predicate does not assert closure, absence, comparison equality, or root minimality, the extension does not invalidate the witness.
\end{proof}

\begin{remark}[Non-monotone gates]
The classes \texttt{TERMINAL\_CLOSURE},\allowbreak \texttt{ABSENCE\_NEGATIVE\_CONTROL},\allowbreak \texttt{RECOMPUTED\_COMPARISON},\allowbreak and \texttt{MULTI\_ROOT\_CORROBORATION} are terminal or recomputed predicates. They must be evaluated again whenever the bounded evidence set or derivation topology changes.
\end{remark}

\subsection{Hash soundness versus semantic completeness}
Collision resistance supports tamper detection relative to a fixed commitment; semantic completeness depends on the schema and actual checks. If a load-bearing field is omitted from canonicalization, the hash can remain stable while the system meaning changes. Mutation testing and independent review are therefore part of the evidence model, not optional QA.

\subsection{Independence is relational}
No layer is independent in isolation. Independence is a relation between derivation paths and failure modes. Two languages can share one faulty specification; two processes can share one stale database; two reviewers can use one verifier. The paper therefore reports independence classes rather than a binary adjective.

\section{Public evidence archive}
The public companion retains sanitized manifests, FIA definitions, schemas, provenance records, checksums, and an offline verifier. The seven named layers are registered derivations whose publication records share accepted cohort metadata; the companion does not publish raw runtime payloads needed to infer seven independent observations. The verifier recomputes repository-file checksums and checks selected required fields, counts, ordering, canonical identifiers, and internal consistency; it cannot reconstruct deliberately unpublished raw roots or independently recompute payload-derived identities. Inspection of the public verifier separates selected structural checks from withheld-runtime attestations; reading that verifier is not execution of the private runtime or full JSON Schema validation.

\section{Relation to assurance and distributed-systems work}
The framework draws on linearizability and transaction-processing ideas \citep{herlihy1990linearizability,gray1993transaction}, causal ordering \citep{lamport1978time}, durable distributed-data design \citep{kleppmann2017ddia}, canonical serialization \citep{rfc8785}, and software supply-chain provenance \citep{torresarias2019intoto,slsa2023,nistssdf2022}. Smart-contract testing research contributes fuzzing, differential analysis, and large-scale tool comparison \citep{wustholz2020harvey,durieux2020smartbugs}. The novel emphasis here is not another testing tool. It is a claim architecture that binds financial authority, runtime identity, cross-layer derivation, reproducible packaging, and promotion language into one reviewable object.

Related work on case-wise conformance and provenance-aware synthetic traces informs the evidence-layer abstraction without contributing an additional cohort or denominator here \citep{aemb2026,avet2026}.

\section{Publication and reproducibility record}
The registered local root topology is the common-cohort case: seven retained positions per FIA are published with registered derivation metadata, while raw derivation payloads are intentionally absent from the public companion. Consequently, the observed result supports cohort consistency and registered provenance, not a claim that seven independent runtime systems observed each FIA. The proposed multi-root corroboration gate was not evaluated for this scope and remains \texttt{NOT\_EVALUATED}.

The public evidence companion retains sanitized manifest, ledger, materialization, package, replay, and binding metadata, together with per-FIA correlation records, schemas, checksums, and an offline metadata verifier. Two clean materializations and the separately verified strict archive provide package-level reproducibility identities; the selected ordered replay provides a registered local replay identity. Raw runtime payloads are excluded, so the public verifier checks the integrity and consistency of distributed metadata rather than reconstructing unpublished raw roots.

The public article reports the evidence method and registered identities. Broader source-code publication, operational deployment, and successor external-operation experiments are governed by separate release and assurance processes.

\section{Limitations}
The framework cannot prove that an oracle is truthful, a venue is solvent, a key is uncompromised, a smart contract has no unknown vulnerability, or a legal claim will be enforced. Multi-root corroboration is a practical, claim-specific property, not metaphysical independence; shared code, libraries, operators, or infrastructure can create common-mode error. Hashes detect changed bytes under specified assumptions, not semantic omissions outside the schema. The current evidence establishes the registered cohort and provenance predicates. Production performance, real custody, external-venue behavior, terminal settlement, and human usability remain separate empirical questions.

\section{Conclusion}
Manifest-bound evidence for financial protocols should be treated as a typed, reproducible proof chain. A receipt, screenshot, test, or copied hash can be useful without proving the whole workflow. The architecture developed here requires separate retained derivations, a shared correlation identity, specified tamper rejection, clean-run reproducibility, historical provenance, and independent review before a claim is promoted. Its most important output is a disciplined, reproducible boundary between what was observed, what was reconstructed, what was replayed, and what remains a future hypothesis.

\section*{Funding}
No external funding was received for this research.

\section*{Competing Interests}
A related pending UAE patent application is disclosed in the Patent Disclosure section.

\section*{Data and Code Availability}
The article reports a registered local engineering-evidence cohort. The retained public evidence companion supports recomputation of its published structure and metadata; it is not an external dataset deposit, a production record, or a live-venue data source.

\noindent\textbf{JEL:} C63, G13, G14, G23.

\section*{AI Disclosure}
OpenAI ChatGPT and Codex were used for editorial and technical assistance during manuscript preparation. The author made all substantive research decisions, reviewed the final manuscript, and assumes full responsibility for its contents.

\section*{Patent Disclosure}
Certain technical aspects described in this article are included in a UAE patent application filed on 30 September 2026 (Application No. MOE-Patents-PTA-0003399-20260930). The application is pending; no granted patent is claimed.

\FloatBarrier
\begingroup
\raggedright
\bibliographystyle{plainnat}
\bibliography{references}

@article{aemb2026,
  author  = {Maksym Nechepurenko},
  title   = {{AEMB: A Deterministic Cross-Language Verification Benchmark for Event-Margin Protocols}},
  journal = {SSRN Electronic Journal},
  year    = {2026},
  month   = {July},
  doi     = {10.2139/ssrn.7216198},
  url     = {https://ssrn.com/abstract=7216198},
  note    = {SSRN abstract 7216198; 22 July 2026; AEMB v0.1.0}
}

@article{avet2026,
  author  = {Maksym Nechepurenko},
  title   = {{AVET: A Provenance-Aware Venue Emulator Trace Dataset for Event Markets}},
  journal = {SSRN Electronic Journal},
  year    = {2026},
  month   = {July},
  doi     = {10.2139/ssrn.7216180},
  url     = {https://ssrn.com/abstract=7216180},
  note    = {SSRN abstract 7216180; 27 July 2026; AVET v0.2.0}
}

@misc{axientfia2026,
  author       = {{AxientLab}},
  title        = {{Axient FIA Public: Reproducibility and Evidence Companion}},
  year         = {2026},
  howpublished = {Public GitHub repository},
  note         = {Public evidence companion; accessed August 2026},
  url          = {https://github.com/AxientLab/axient-fia-public/tree/9a473fbd49cbd4f95c64bb0bd090bbdaeac3f442}
}

@article{herlihy1990linearizability,
  author  = {Maurice P. Herlihy and Jeannette M. Wing},
  title   = {Linearizability: A Correctness Condition for Concurrent Objects},
  journal = {ACM Transactions on Programming Languages and Systems},
  volume  = {12}, number = {3}, pages = {463--492}, year = {1990},
  doi     = {10.1145/78969.78972},
  url     = {https://doi.org/10.1145/78969.78972}
}

@book{gray1993transaction,
  author    = {Jim Gray and Andreas Reuter},
  title     = {Transaction Processing: Concepts and Techniques},
  publisher = {Morgan Kaufmann},
  year      = {1993}
}

@article{lamport1978time,
  author  = {Leslie Lamport},
  title   = {Time, Clocks, and the Ordering of Events in a Distributed System},
  journal = {Communications of the ACM},
  volume  = {21}, number = {7}, pages = {558--565}, year = {1978},
  doi     = {10.1145/359545.359563},
  url     = {https://doi.org/10.1145/359545.359563}
}

@book{kleppmann2017ddia,
  author    = {Martin Kleppmann},
  title     = {Designing Data-Intensive Applications},
  publisher = {O'Reilly Media},
  year      = {2017}
}

@techreport{rfc8785,
  author      = {Anders Rundgren and Bret Jordan and Samuel Erdtman},
  title       = {{JSON Canonicalization Scheme (JCS)}},
  institution = {Internet Engineering Task Force},
  number      = {RFC 8785},
  year        = {2020},
  doi         = {10.17487/RFC8785},
  url         = {https://doi.org/10.17487/RFC8785}
}

@inproceedings{torresarias2019intoto,
  author    = {Santiago Torres-Arias and Hammad Afzali and Trishank Karthik Kuppusamy and Reza Curtmola and Justin Cappos},
  title     = {{in-toto}: Providing Farm-to-Table Guarantees for Bits and Bytes},
  booktitle = {28th USENIX Security Symposium},
  pages     = {1393--1410},
  year      = {2019},
  url       = {https://www.usenix.org/conference/usenixsecurity19/presentation/torres-arias},
  note      = {Accessed 12 August 2026}
}

@misc{slsa2023,
  author = {{Open Source Security Foundation}},
  title  = {{Supply-chain Levels for Software Artifacts (SLSA) Specification, Version 1.0}},
  year   = {2023},
  url    = {https://slsa.dev/spec/v1.0/},
  note   = {Final specification, 19 April 2023; accessed 12 August 2026}
}

@techreport{nistssdf2022,
  author      = {Murugiah Souppaya and Karen Scarfone and Donna Dodson},
  title       = {{Secure Software Development Framework (SSDF) Version 1.1}},
  institution = {National Institute of Standards and Technology},
  number      = {NIST SP 800-218},
  year        = {2022},
  doi         = {10.6028/NIST.SP.800-218},
  url         = {https://doi.org/10.6028/NIST.SP.800-218}
}

@inproceedings{claessen2000quickcheck,
  author    = {Koen Claessen and John Hughes},
  title     = {{QuickCheck}: A Lightweight Tool for Random Testing of {Haskell} Programs},
  booktitle = {Proceedings of the Fifth ACM SIGPLAN International Conference on Functional Programming},
  pages     = {268--279}, year = {2000},
  doi       = {10.1145/351240.351266},
  url       = {https://doi.org/10.1145/351240.351266}
}

@article{mckeeman1998differential,
  author  = {William M. McKeeman},
  title   = {Differential Testing for Software},
  journal = {Digital Technical Journal},
  volume  = {10}, number = {1}, pages = {100--107}, year = {1998}
}

@article{jia2011mutation,
  author  = {Yue Jia and Mark Harman},
  title   = {An Analysis and Survey of the Development of Mutation Testing},
  journal = {IEEE Transactions on Software Engineering},
  volume  = {37}, number = {5}, pages = {649--678}, year = {2011},
  doi     = {10.1109/TSE.2010.62},
  url     = {https://doi.org/10.1109/TSE.2010.62}
}

@inproceedings{wustholz2020harvey,
  author    = {Valentin W\"ustholz and Maria Christakis},
  title     = {{Harvey}: A Greybox Fuzzer for Smart Contracts},
  booktitle = {Proceedings of the 28th ACM Joint European Software Engineering Conference and Symposium on the Foundations of Software Engineering},
  pages     = {1398--1409}, year = {2020},
  doi       = {10.1145/3368089.3417064},
  url       = {https://doi.org/10.1145/3368089.3417064}
}

@inproceedings{durieux2020smartbugs,
  author    = {Thomas Durieux and {Jo{\~a}o F. Ferreira} and Rui Abreu and Pedro Cruz},
  title     = {Empirical Review of Automated Analysis Tools on 47,587 {Ethereum} Smart Contracts},
  booktitle = {Proceedings of the 42nd International Conference on Software Engineering},
  pages     = {530--541}, year = {2020},
  doi       = {10.1145/3377811.3380364},
  url       = {https://doi.org/10.1145/3377811.3380364}
}
\endgroup
\end{document}